\documentclass[runningheads,a4paper]{llncs}
\usepackage{amssymb,amsmath,pifont,footnote}
\usepackage{latexsym,pdfpages}
\usepackage{array,verbatim}
\usepackage{stmaryrd,ifsym}
\usepackage{algorithmic,algorithm}
\usepackage{cite}

\newtheorem{thm}{Theorem} 

\newtheorem{lem}{Lemma}

\newtheorem{prop}{Proposition}
\newtheorem{cor}{Corollary}

\newcommand{\pfbox}{\qed}

\newcommand{\Q}{\mathbb{Q}}

\newcommand{\Heading}[1]{\vspace{-0.25cm}\paragraph{\bf{#1}}}
\newcommand{\BeginProof}{\vspace{-0.25cm}\begin{proof}}

\newcommand{\Comment}[1]{}
\newcommand{\Appendix}[1]{}

\newcommand{\Avg}{\mathit{Avg}}

\newcommand{\VEC}[1]{\ensuremath{\overline{#1}}}

\newcommand{\Nat}{\ensuremath{\mathbb{N}}}

\newcommand{\OPT}{\mathit{OPT}}

\begin{document}

\pagestyle{plain}

\newcommand{\AppendixString}{\cite{mpexpressionCorr}}
\newcommand{\AppendixContent}[1]{}

\newcommand{\MPG}{\mathit{MPG}}

\title{The Complexity of\\Infinitely Repeated Alternating Move Games}

\author{Yaron Velner}
\institute{The Blavatnik School of Computer Science, Tel Aviv University, Israel}
\maketitle

\begin{abstract}
We consider infinite duration alternating move games.
These games were previously studied by Roth, Balcan, Kalai and Mansour~\cite{mansour}.
They presented an FPTAS for computing an approximate equilibrium, and conjectured that there is a polynomial algorithm for finding an exact equilibrium~\cite{mansourInternet}.
We extend their study in two directions:
(1)~We show that finding an exact equilibrium, even for two-player zero-sum games, is polynomial time equivalent to finding a winning strategy for a (two-player) mean-payoff game on graphs. The existence of a polynomial algorithm for the latter is a long standing open question in computer science.
Our hardness result for two-player games suggests that two-player alternating move games are harder to solve than two-player simultaneous move games,
while the work of Roth et al., suggests that for $k\geq 3$, $k$-player games are easier to analyze in the alternating move setting.
(2)~We show that optimal equilibria (with respect to the social welfare metric) can be obtained by pure strategies, and we present an FPTAS for computing a pure approximated equilibrium that is $\delta$-optimal with respect to the social welfare metric.
This result extends the previous work by presenting an FPTAS that finds a much more desirable approximated equilibrium.
We also show that if there is a polynomial algorithm for mean-payoff games on graphs, then there is a polynomial algorithm that computes an optimal exact equilibrium, and hence, (two-player) mean-payoff games on graphs are inter-reducible with $k$-player alternating move games, for any $k\geq 2$.
\end{abstract}

\section{Introduction}
In this work, we investigate infinitely repeated games in which players \emph{alternate} making moves.
This framework can model, for example, five telecommunication providers competing for customers: each company can observe the price that is set by the others, and it can update the price at any time.
In the short term, each company can benefit from undercutting its opponents price, but since the game is repeated indefinitely, in some settings, it might be better to coordinate prices with the other companies.
Such examples motivate us to study equilibria in alternating move games.

In this work, we study infinitely repeated $k$-player $n$-action games.
In such games, in every \emph{round}, a player chooses an action, and the utility of each player (for the current round) is determined according to the $k$-tuple of actions of the players.
Each player goal is to maximize his own \emph{long-run average} utility as the number of rounds tends to infinity.

These games were studied by Roth et al. in~\cite{mansour}, and they showed an FPTAS for computing an $\epsilon$-equilibrium.
Their result provided a theoretical separation between the alternating move model and the simultaneous move model, since for the latter, it is known that an
FPTAS for computing approximate equilibria does not exists for games with $k\geq 3$ players unless P=PPAD.
Their result was obtained by a simple reduction to mean-payoff games on graphs.
These games were presented in~\cite{mycielski}, and they play an important rule in automata theory and in economics.
The computational complexity of finding an exact equilibrium for such games is a long standing open problem, and despite many efforts~\cite{Gurvich198885,Zwick,Andrews:1999,MP,Bjdrklund2007210,Doyen}, there is no known polynomial solution for this problem.
 
We extend the work in~\cite{mansour} by investigating the complexity of an exact equilibrium (which was stated as an open question in~\cite{mansour}), and by investigating the computational complexity of finding an \emph{$\delta$-optimal} approximated equilibrium with respect to the \emph{social welfare} metric.
Our main technical results are as follows:
\begin{itemize}
\item We show a reduction from mean-payoff games on graphs to two-player zero-sum alternating move games, and thus we prove that $k$-player alternating move games are computationally equivalent to mean-payoff games on graphs for any $k\geq 2$.
\item We show that optimal equilibrium can be obtained by pure strategies, and
we show an FPTAS for computing an $\delta$-optimal $\epsilon$-equilibrium.
In addition, we show that computing an exact optimal equilibrium is polynomial time equivalent to solving mean-payoff games on graphs.
\end{itemize}
We note that the first result may suggest that two-player alternating move games are harder than two-player simultaneous move games, since a polynomial time algorithm to solve the latter is known~\cite{Littman04apolynomial}.
Hence, along with the result of~\cite{mansour}, we get that simultaneous move games are easier to solve with comparison to alternating move games for two-player games, and are harder to solve for $k\geq 3$ player games.

This paper is organized as following.
In the next section we bring formal definitions for alternating move games and mean-payoff games on graphs.
In Section~\ref{sect:LowerBound}, we show that alternating move games are at least as hard as mean-payoff games on graphs.
In Section~\ref{sect:Optimal} we investigate the properties of optimal equilibria, and we present an FPTAS for computing an $\delta$-optimal $\epsilon$-equilibrium. 
Due to lack of space, some of the proofs were omitted, and the full proofs are given in the appendix.
\section{Definitions}
In this section we bring the formal definitions for alternating move repeated games and mean-payoff games on graphs.
Alternating move games are presented in Subsection~\ref{subsec:1},
and mean-payoff games are presented in Subsection~\ref{subsec:2}.
\subsection{Alternating move repeated games}\label{subsec:1}
\Heading{Actions, plays and utility function.}
A $k$-player $n$-action game is defined by an action set $A_i$ for every player $i$, and by $k$ utility functions, one for each player, $u_i : A_1 \times \dots A_k \to [-1,1]$.
W.l.o.g we assume that the size of all action set is the same, and we denote it by $n$.
We note that any game can
be rescaled so its utilities are bounded in $[-1, 1]$, however, the FPTAS that was presented in~\cite{mansour}, and our results in Subsection~\ref{subsec:Approx} crucially rely on the assumption that the utilities are in the interval $[-1,1]$.

An alternating move game is played for infinitely many rounds.
In round $t$ player $j = 1 + (t\mod k)$ plays action $a_j^t$, and a \emph{vector of actions} $a^t = (a_1,\dots,a_k)$ is produced, where $a_i\in A_i$ is the last action of player $i$.
In every round $t$, player $i$ receives a utility $u_i(a^t)$, which depends only in the last action of each of the $k$ players (W.l.o.g the utility in the first $k$ rounds is zero for all players).
A sequence of infinite rounds forms a $\emph{play}$, and we characterize a play either by an infinite sequence of actions or by the corresponding sequence of vectors of actions.
The utility of player $i$ in a play $a^1 a^2 \dots a^t a^{t+1} \dots$ is the \emph{limit average payoff}, namely, $\lim_{n\to\infty} \frac{1}{n}\sum_{t=1}^n u_i(a^t)$.
When this limit does not exist, we define the utility of the play for player $i$ to be $\liminf_{n\to\infty} \frac{1}{n}\sum_{t=1}^n u_i(a^t)$.
We note that in the frame work of~\cite{mansour}, the utility of a play was undefined when the limit does not exist.
The results we present in this paper for the $\liminf$ metric holds also for the framework of~\cite{mansour}.
On the other hand, if we would take the $\limsup$ value instead, then the problem is much easier.
An optimal equilibrium is obtained when all players join forces and maximize player~$1 + (i \pmod k)$ utility for $2^i$ rounds (for $i=1,2,\dots,\infty$).
Hence, we can easily produce a polynomial algorithm to solve these games.
\Heading{Strategies.}
A \emph{strategy} is a recipe for player's next action, based on the entire \emph{history} of previous actions.
Formally, a (mixed) strategy for player $i$ is a function
$\sigma_i : (A_1\times \dots \times A_k)^*  \times A_1\times\dots \times A_{i-1} \to \Delta(A_i)$,
where $\Delta(S)$ denotes the set probability distribution over any finite set $S$.
We say that $\sigma_i$ is a \emph{pure strategy} if $\Delta$ is a degenerated distribution.
A \emph{strategy profile} is a vector $\sigma = (\sigma_1,\dots,\sigma_k)$ that defines a strategy for every player.
A profile of pure strategies uniquely determines the action vector in every round and yields a utility vector for the players.
A profile of mixed strategies determines, for every round $t$ in the play, a distribution of sequences of action vectors, and the \emph{average payoff} in round $t$ is the expected average payoff over the distribution of action vectors.
Formally, for a strategy profile $\sigma$ we denote the average payoff of player $i$ in round $t$ by
\[P_{i,t}(\sigma) = E_{a^1 a^2 \dots a^t \sim \sigma}[\frac{u_i(a^1) + \dots + u_i(a^t)}{t}]\]
and the utility of player $i$ is $\liminf_{t\to\infty} P_{i,t}$.

\Heading{Equilibria, $\epsilon$ equilibria and optimal equilibria}
A strategy profile forms an \emph{equilibrium} if none of the players can strictly improve his utility (that is induced by the profile) by unilaterally deviating from his strategy (that is defined by the profile).
For every $\epsilon > 0$, we say that a strategy profile forms an \emph{$\epsilon$-equilibrium} if none of the players can improve his utility by more than $\epsilon$ by unilaterally deviating from his strategy.

The \emph{social welfare} of a strategy profile is the sum of the utilities of all players.
An equilibrium (resp. an $\epsilon$-equilibrium) is called an \emph{optimal equilibrium} (optimal $\epsilon$-equilibrium) if its social welfare is not smaller than the social welfare of any other equilibrium ($\epsilon$-equilibrium).
For $\delta > 0$, an equilibrium (resp. an $\epsilon$-equilibrium) is called an \emph{$\delta$-optimal equilibrium} if its social welfare is not smaller by more than $\delta$ with comparison to the social welfare of any other equilibrium ($\epsilon$-equilibrium).

\subsection{Mean-payoff games on graphs}\label{subsec:2}
\Heading{Plays and payoffs.}
A \emph{mean-payoff game} on a graph is defined by a weighted directed bipartite graph
$G=(V = V_1\cup V_2, E, w: E \to \Q)$ and an initial vertex $v_0 \in V$.
The game consists of two players, namely, maximizer (who owns $V_1$) and minimizer (who owns $V_2$).
Initially, a pebble is place on the initial vertex, and in every round, the player who owns the vertex in which the pebble resides, advance the pebble into an adjacent vertex.
This process is repeated forever and forms a \emph{play}.
A play is characterized by a sequence of edges, and the
\emph{average payoff} of a play $\rho = e_1 \dots e_t$ up to round $t$ is denoted by
$P_t = \frac{1}{t} \sum_{i=1} w(e_i)$.
The \emph{value} of a play is the limit average payoff (mean-payoff), namely,
$\liminf_{t\to\infty} P_t$.
(We note that for games on graphs, the $\limsup$ metric gives the same complexity results.)
The objective of the maximizer is to maximize the mean-payoff of a play, and the minimizer aims to minimize the mean-payoff.

\Heading{Strategies, memoryless strategies, optimal strategies and winning strategies.}
In this work, we consider only pure strategies for games on graphs, and it is well-known that randomization does not give better strategies for mean-payoff games.
A \emph{strategy} for maximizer is a function $\sigma : (V_1 \times V_2)^*\times V_1 \to E$ that decides the next move, and similarly, for the minimizer a strategy is a function $\tau : (V_1 \times V_2)^* \to E$.
A strategy is called \emph{memoryless} if it depends only on the current position of the pebble.
Formally, a memoryless strategy for the maximizer is a function $\sigma : V_1 \to E$ and similarly a memoryless strategy for the minimizer is a function $\tau : V_2 \to E$.

A profile of strategies $(\sigma,\tau)$ uniquely determines the mean-payoff value of a game.
We say that a play $\pi = e_1 e_2 \dots e_n \dots$ is \emph{consistent} with a maximizer strategy $\sigma$ if there exists a minimizer strategy $\tau$ such that $\pi$ is formed by $(\sigma,\tau)$.
We say that the \emph{value} of a maximizer strategy is $p$ if it can assure a value of at least $p$ against any minimizer strategy.
Analogously, we say that the value of a minimizer strategy is $p$ if it can assure a value of at most $p$ against any maximizer strategy.

We say that a maximizer strategy is \emph{optimal} if its value is maximal (with respect to all possible maximizer strategies).
Analogously, a minimizer strategy is optimal if its value is minimal.
For a given threshold, we say that a maximizer strategy is a \emph{winning strategy} if it assures mean-payoff value that is greater or equal to the given threshold, and a minimizer strategy is winning if it assures value that is 
strictly smaller than the given threshold.

\Heading{One-player games, and games according to memoryless strategies}
A special (and easier) case of games on graphs is when the out-degree is one for all the vertices that are owned by a certain player.
In this case, all the \emph{choices} are done by one player.
For a two-player game on graph $G$, and a player-1 strategy $\sigma$, we define the one-player game graph $G^\sigma$ to be the game graph that is formed by removing, for every player-1 vertex $v$, the out-edges that are not equal to $\sigma(v)$.

\Heading{Classical results on mean-payoff games.}
Mean-payoff games were introduced in '79 by Ehrenfeucht and Mycielski~\cite{mycielski}, and their main result was that optimal strategies (for both players) exist, and moreover, the optimal value can be obtained by a memoryless strategy.
The decision problem for mean-payoff games is to determine if the maximizer has a winning strategy with respect to a given threshold.
The existence of optimal memoryless strategies almost immediately proves that the decision problem for mean-payoff games is in NP$\cap$coNP, and thus it is unlikely to be NP-hard (or coNP-hard). 
Zwick and Paterson~\cite{Zwick} introduced the first pseudo-polynomial algorithm, which runs in polynomial time when the weights of the edges are encoded in unary.
They also provided a polynomial algorithm for the special case of one-player mean-payoff games.
A randomized sub-exponential algorithm for mean-payoff games is also known~\cite{Bjdrklund2007210}, but despite many efforts, the existence of a polynomial algorithm to solve mean-payoff games remains an open question, and it is one of the rare problems in computer science that is known to be in NP$\cap$coNP but no polynomial algorithm is known.

We summarize the known results on mean-payoff games in the next theorems.
The first theorem states that optimal strategies exist and moreover, there exist optimal strategies that are memoryless.
\begin{thm}[\cite{mycielski}]\label{thm:OptStrategies}
For every mean-payoff game there exists a maximizer memoryless strategy $\sigma$ and a minimizer memoryless strategy $\tau$ such that $\sigma$ is optimal for the maximizer and $\tau$ is optimal for the minimizer.
\end{thm}
The next theorem shows that there is a polynomial algorithm that computes optimal strategies if and only if there is a polynomial algorithm for the mean-payoff games decision problem.
\begin{thm}[\cite{Zwick}]\label{thm:DecAndCompTheSame}
The following problems are polynomial time inter-reducible:
(i)~Compute maximizer optimal memoryless strategy.
(ii)~Compute the optimal value that maximizer can assure.
(iii)~Determine whether maximizer optimal value is at least zero.
(iv)~Determine whether maximizer optimal value is greater than zero.
\end{thm}

\section{Two-Player Zero-Sum (Alternating Move) Games are Inter-Reducible with Mean-Payoff Games}\label{sect:LowerBound}
In this section we prove that there is a polynomial algorithm that computes an exact equilibrium for two-player zero-sum (alternating move) games if and only if there exists a polynomial algorithm that solves mean-payoff games.

The reduction from two-player zero-sum games to mean-payoff games is trivial, 
for a two-player zero-sum game with actions $A_1$, $A_2$ and utility functions $u_1 : A_1 \times A_2 \to \Q$ and $u_2 = -u_1$, we construct a complete bipartite game graph $G = (V=A_1 \cup A_2, E = (A_1 \times A_2) \cup (A_2 \times A_1),w:E\to \Q)$ such that the weight of the transition from $a_1\in A_1$ to $a_2\in A_2$ is simply $u_1(a_1,a_2)$, and the weight of the transition from $a_2$ to $a_1$ is also $u_1(a_1,a_2)$.
It is a simple observation that a pair of optimal strategies (for the maximizer and minimizer) in the mean-payoff game induces an equilibrium strategy profile in the two-player zero-sum game and vice versa.

The reduction for the converse direction is more complicated.
For this purpose we bring the notion of \emph{undirected game graph}.
A mean-payoff game graph is said to be \emph{undirected} if its edge relation is symmetric, and $w(v_1,v_2) = w(v_2,v_1)$ for every edge $(v_1,v_2)$.
(Basically, it is a game on an undirected graph.)
The next simple lemma shows a reduction from mean-payoff games on a complete bipartite undirected graphs to two-player zero-sum game.
\begin{lem}\label{lem:ForGameToUnderictedComplete}
There is a polynomial reduction from mean-payoff games on complete bipartite undirected graphs to two-player zero-sum games.
\end{lem}
\begin{proof}
The proof is straight forward.
Let $V_1$ and $V_2$ be the maximizer and minimizer (resp.) vertices in the mean-payoff game.
We construct a two-player zero-sum game in the following way.
The set of action of player 1 is $A_1 = V_2$ and the set of action for player 2 is $A_2 = V_1$.
We denote by $W$ the least value for which all the weights in the undirected graphs are in $[-W,+W]$,
and the utility function of player 1 is $u_1(a_1,a_2) = \frac{w(a_1,a_2)}{W} = \frac{w(a_2,a_1)}{W}$, and $u_2 = -u_1$.
It is trivial to observe that an equilibrium profile induces a pair of optimal strategies for the mean-payoff game, and the proof follows. \pfbox
\end{proof}

Due to Lemma~\ref{lem:ForGameToUnderictedComplete}, all that is left is to prove that mean-payoff games on complete bipartite undirected graphs are equivalent to mean-payoff games.
A recent result by Chatterjee, Henzinger, Krinninger and Nanongkai~\cite{CH12} gives us the first step towards such proof.
\begin{thm}[Corollary~24 in \cite{CH12}]\label{thm:CH}
Solving mean-payoff games on complete bipartite (directed) graphs is as hard as solving mean-payoff games on arbitrary graphs.
\end{thm}
We use the above result as a black box and extend it to complete bipartite undirected graphs.
We note that the main difference between directed and undirected graphs is that for undirected graphs the weight function is symmetric.
In the rest of this section we will describe a process that for a given complete bipartite directed graph, generates a suitable symmetric weight function, and the winner in the generated graph is the same as in the original graph.

We say that a directed game graph has a \emph{normalized weight function} if it assigns a positive weight to every out-edge of maximizer vertex, and a negative weight for every out-edge of minimizer vertex.
The next lemma shows that we may assume w.l.o.g that a directed game graph has a normalized weight function.
\begin{lem}\label{lem:NormWeightFunction}
Solving mean-payoff games on (directed) bipartite graphs is polynomial time inter-reducible to solving mean-payoff games on (directed) bipartite graphs with normalized weights.
\end{lem}
\begin{proof}
Let $G$ be a non-normalized graph and let us denote by $W$ the heaviest weight (in absolute value) that is assigned by its weight function $w$.
We construct a normalized graph $G'$ from $G$ by defining a weight function $w'$ as:
\[
w'(u,v) = \left\{ \begin{array}{rl}
w(u,v) + (W + 1) &\mbox{ if $u$ is owned by maximizer} \\
w(u,v) - (W + 1) &\mbox{ otherwise (if $u$ is owned by the minimizer)}
\end{array} \right.
\]
Clearly, $G'$ is a normalized graph, and since $G'$ and $G$ are bipartite, it is straight forward to observe that for any finite path in $\pi$ we have that $|w(\pi) - w'(\pi)| \leq W+1$, and thus, for every infinite path $\rho$, we have that that mean-payoff value of $\rho$ according to $w$ is identical to the mean-payoff of $\rho$ according to $w'$.
Therefore, a maximizer winning (resp. optimal) strategy in graph $G$ is a winning (optimal) strategy also in $G'$ and vice versa, and the proof of the lemma follows. \pfbox
\end{proof}

In the next lemma we show that mean-payoff games on direct normalized bipartite complete graphs are as hard as mean-payoff games on undirected normalized bipartite complete graphs.
\begin{lem}\label{lem:DirectAndUndirect}
The problem of determining whether maximizer has a winning strategy for a threshold $0$ for
a mean-payoff games on a directed normalized bipartite complete graph is as hard as the corresponding problem for mean-payoff games on an undirected normalized bipartite complete graph
\end{lem}
To conclude, by Lemma~\ref{lem:ForGameToUnderictedComplete} we get that a polynomial algorithm for alternating move two-player zero-sum games exists if and only if there exists a polynomial algorithm for solving mean-payoff games on a complete bipartite undirected graph,
and by Lemmas~\ref{lem:NormWeightFunction},~\ref{lem:ForGameToUnderictedComplete} and~\ref{lem:DirectAndUndirect}
and by Theorem~\ref{thm:CH} we get that the latter exists if and only if there exists a polynomial algorithm for solving mean-payoff games on arbitrary (directed) graphs.
Hence, the main result of this section follows.
\begin{thm}\label{thm:Sec1}
There exists a polynomial time algorithm for computing exact equilibrium for two-player zero-sum (alternating move) games if and only if there exists a polynomial time algorithm for solving mean-payoff games on graphs.
\end{thm}

\section{Complexity of Computing Optimal Equilibrium}\label{sect:Optimal}
In this section, we investigate the complexity of computing an optimal equilibrium.
Our main results are summarized in the next theorem:
\begin{thm}\label{thm:OptEqilib}
\begin{enumerate}
\item Optimal equilibrium can be obtained by a profile of pure strategies.
\item If mean-payoff games are in P, then there is a polynomial algorithm for computing an exact optimal equilibrium.
\item If mean-payoff games are not in P, then there is no FPTAS that approximate the social welfare of the optimal equilibrium.
\item There is an FPTAS to compute an $\epsilon$-equilibrium that is $\delta$-optimal.
(Note that it does not necessarily approximate the value of an exact optimal equilibrium.)
\end{enumerate}
\end{thm}

We will prove Theorem~\ref{thm:OptEqilib} in the next four subsections:
In Subsection~\ref{subsec:naive} we show the naive algorithm for computing an equilibrium that is based on Folk Theorem, and we prove basic properties of equilibria in alternating move games.
In Subsection~\ref{subsec:Pure} we prove Theorem~\ref{thm:OptEqilib}(1).
In Subsection~\ref{subsec:ComValue} we investigate the complexity of computing the social welfare of the optimal equilibrium, and prove Theorem~\ref{thm:OptEqilib}(2) and Theorem~\ref{thm:OptEqilib}(3).
Finally, in Subsection~\ref{subsec:Approx} we prove Theorem~\ref{thm:OptEqilib}(4) which is the main result of this section.

In this section, we will model $n$-action $k$-player alternating move games by a multi-weighted graph, according to the following conventions:
The vertices of the graph are the vertices in the set
$V=(A_1\times A_2 \times \dots \times A_k)\times\{1,\dots,k\}$,
and we say that player $i$ owns the vertex set $V_i = (A_1\times A_2 \times \dots \times A_k) \times \{i\}$.
Intuitively, a vertex is characterized by an action vector and by a player that owns it.
The pair $(u,v)$ is in the edge relation if $u$ is owned by player $i$, $v$ is owned by player $i+1$ (where player $k+1$ is player 1), and there is at most one difference in the action vector of $u$ and $v$ and it is in position $i$.
The weight of every edge is a vector of size $k$ that corresponds to the utility vector of the actions.
Formally, if $u = (\VEC{a_1},i)$ and $v=(\VEC{a_2},i+1)$ then $w(u,v) = (u_1(\VEC{a_2}), u_2(\VEC{a_2}), \dots, u_k(\VEC{a_2}))$.

For an infinite path in the multi-weighted graph we define the dimension $i$ of \emph{mean-payoff vector} of the path to be the mean-payoff value of the path according to dimension $i$.
It is an easy observation that every infinite path in the graph corresponds to a play and its mean-payoff vector corresponds to the utility vector of the play.
We note that the size of the graph is $k^2\cdot n^k$ which is polynomial in the size of the encoding of the utility functions (which is $k\cdot n^k$), hence this graph can be constructed in polynomial time.

\subsection{Basic properties of equilibria}\label{subsec:naive}
The Folk Theorem gives a conceptually simple (but inefficient) technique to construct an equilibrium.
Intuitively, an equilibrium is obtained when each of the players play as if the goal of all the other players is to minimize its utility, and if one of the  players deviates from this strategy, then all the other players switch to playing according to a strategy that will minimize the utility of the rebellious player.
Formally, let $G$ be the corresponding $k$-player game graph that models the alternating move game.
For every player $i$, we consider a zero-sum two-player mean-payoff game graph $G^i$ in which the maximizer owns player $i$ vertices and the minimizer owns the other vertices.
Let $\sigma_i$ be an arbitrary optimal strategy for the maximizer in graph $G^i$, let $\overline{\sigma_i}$ be an arbitrary optimal strategy for the minimizer in $G^i$, and let $\nu_i$ be the value that is obtained by the strategy profile $(\sigma_i,\overline{\sigma_i})$.
Then if every player $i$ plays according to the strategy:
\begin{quote}
If player $j\neq i$ deviated from $\sigma_j$, then play according to $\overline{\sigma_j}$ forever, and otherwise play according to $\sigma_i$
\end{quote}
an equilibrium is formed (since by definition, playing according to $\sigma_i$ assures utility at least $\nu_i$, and deviating from $\sigma_i$ assures utility at most $\nu_i$).

In the next lemma, we extend the basic principle of Folk Theorem and get a characterization of all the equilibria that are obtained by a profile of pure strategies.
\begin{lem}\label{lem:PureStrategies}
Let $(t_1,t_2,\dots,t_k)$ be a utility vector such that $t_i \geq \nu_i$ (for every player $i$), then there exists a pure equilibrium with utility exactly $t_i$ for every player $i$ if and only if there exists an infinite path $\pi$ in the graph $G$ with mean-payoff vector $(t_1,t_2,\dots,t_k)$.
\end{lem}

\subsection{Optimal equilibrium can be obtained by pure strategies}\label{subsec:Pure}
In this subsection, we extend Lemma~\ref{lem:PureStrategies} also for the case of mixed strategies, and as a consequence we get that optimal equilibrium can be obtained by pure strategies.
Intuitively, we wish to show that if a profile of (mixed) strategies yields a utility vector $(t_1,\dots,t_k)$, then there exists an infinite path in the graph with mean-payoff vector that is greater or equal (in every dimension) to $(t_1,\dots,t_k)$.
Then we get that if a utility vector is obtained by a profile of mixed strategies, and then by Lemma~\ref{lem:PureStrategies} it is also obtained by a profile of pure strategies.

We formally prove the above by the next two lemmas.
\begin{lem}\label{lem:SufCond}
Let $G$ be a multi-weighted graph that is strongly connected, and let $(t_1,\dots,t_k)$ be a vector.
Then if for every $\alpha > 0$ there exists a (finite) cyclic path with average weight at least $t_i - \alpha$ in every dimension, then there exists an infinite path with mean-payoff vector at least $(t_1,\dots,t_k)$.
\end{lem}
\begin{lem}\label{lem:MixedFollowsCycle}
Let $\sigma$ be a profile of (mixed) strategies with utility vector $(t_1,\dots,t_k)$.
Then for every $\alpha > 0$ there exists a cyclic path in the game graph with average weight at least $t_i - \alpha$ in every dimension.
\end{lem}
We are now ready to prove that the utility vector of a mixed equilibrium can be obtained by a pure equilibrium.
\begin{prop}\label{prop:MixedNotStrongetThenPure}
Let $\sigma$ be a profile of mixed strategies that induces a utility vector $(t_1,\dots,t_k)$.
Then there exists a profile $\sigma '$ of pure strategies that induces exactly the same utility vector.
Moreover, if $\sigma$ is an equilibrium, then so is $\sigma '$.
\end{prop}
\begin{proof}
By Lemma~\ref{lem:MixedFollowsCycle} we get that for every $\alpha > 0$ there is a cyclic path with average weight at least $t_i - \alpha$ in every dimension.
Therefore, by Lemma~\ref{lem:SufCond} we get that there is an infinite path in $G$ with mean-payoff vector at least $(t_1,\dots,t_k)$, and by Lemma~\ref{lem:PureStrategies} we get that there is a profile of pure strategies that has utility at least $(t_1,\dots,t_k)$.
If $\sigma$ is an equilibrium we get that $t_i \geq \nu_i$ (since otherwise, player $i$ would deviate to strategy $\sigma_i$), and thus, by Lemma~\ref{lem:PureStrategies}, we get that there is a pure equilibrium that gives the same utility vector. \pfbox
\end{proof}
The next corollary immediately follows from Proposition~\ref{prop:MixedNotStrongetThenPure}.
\begin{cor}[Theorem~\ref{thm:OptEqilib}(1)]\label{cor:OptSocialWelfareIsPure}
An optimal equilibrium can be obtained by a profile of pure strategies.
\end{cor}

\subsection{The complexity of computing the social welfare of the optimal (exact) equilibrium}\label{subsec:ComValue}
In this section we show that if there is a polynomial algorithm for mean-payoff games, then there is a polynomial algorithm to compute an optimal equilibrium in a $k$-player alternating move games.
We also prove the converse direction, that is, we show that if there is a polynomial algorithm that computes the social welfare of the optimal equilibrium, then there is a polynomial algorithm that solves mean-payoff games.
We prove these two assertions in the next two lemmas.
\begin{lem}\label{lem:MeanPayoffImpliesOptimal}
Suppose that mean-payoff games are in P, then there is a polynomial algorithm that computes the social welfare of an optimal equilibrium.
\end{lem}
\begin{proof}
Due to Corollary~\ref{cor:OptSocialWelfareIsPure}, it is enough to consider only pure strategies, and due to Lemma~\ref{lem:PureStrategies} a vector of utilities $(t_1,\dots,t_k)$ is obtained by a pure equilibrium if and only if $t_i\geq \nu_i$ (for $i=1,\dots,k$) and there is an infinite path in the game graph with mean-payoff $(t_1,\dots,t_k)$.
Since we assume that there is a polynomial algorithm for computing $\nu_i$,
our problem boils down to 
\begin{quote}
Find the maximal value of $\sum_{i=1}^k t_i$ subject to
\begin{itemize}
\item $t_i \geq \nu_i$; and
\item there exists an infinite path with mean-payoff vector at least $(t_1,\dots,t_k)$
\end{itemize}
\end{quote}
It was shown in~\cite{UmmelsW11} (in the proof of 
Theorem~18) that the problem of deciding whether there exists an infinite path with mean-payoff vector at least $(t_1,\dots,t_k)$ can be reduced (in polynomial time) to a set of linear constraints.
Moreover, the generated set of constraints remain linear even when $t_i$ is a variable.
Hence, we can find a feasible threshold vector $(t_1,\dots,t_k)$ (that is, a vector that is realizable by an infinite path in the graph) that maximizes $\sum_{i=1}^k t_i$ by linear programming.
Therefore, if we have a polynomial algorithm that computes $\nu_i$, then we can find the social welfare of the optimal equilibrium in polynomial time. \pfbox
\end{proof}
Lemma~\ref{lem:MeanPayoffImpliesOptimal} proves Theorem~\ref{thm:OptEqilib}(2) and gives an upper bound to the complexity of computing optimal equilibrium.
In the next lemma we show that this bound is tight, and that the social welfare of the optimal equilibrium cannot be approximated, unless mean-payoff games are in P.
\begin{lem}[Theorem~\ref{thm:OptEqilib}(3)]\label{lem:SocialLowerBound}
There is no FPTAS that approximates the social welfare of an optimal equilibrium, unless mean-payoff games are in P.
\end{lem}

\subsection{An FPTAS to compute an $\epsilon$-equilibrium that is $\delta$-optimal}\label{subsec:Approx}
In this subsection, we assume that the utilities of the players are scaled to rationals in $[-1,1]$, and we will describe an algorithm that computes an $\epsilon$-equilibrium that is $\delta$-optimal (with respect to all $\epsilon$-equilibria) and runs in time complexity that is polynomial in the input size and in $\frac{1}{\epsilon}$ and $\frac{1}{\delta}$.
Subsection~\ref{subsec:ComValue} suggests that in order to compute a $\delta$-optimal $\epsilon$-equilibrium we should approximate (by some value) the values of $\nu_1,\dots,\nu_k$ and then compute the optimal infinite path (with respect to the sum of utilities) that has utility for player $i$ that is greater than the approximation of $\nu_i$.
However, this approach would not work, since the optimal social welfare is not a continuous function with respect to the values $\nu_1,\dots,\nu_k$.

We denote by $\OPT_\epsilon$ the social welfare of the optimal $\epsilon$-equilibrium.
We base our solution on the next lemma, which gives two key properties of $\OPT_\epsilon$.
\begin{lem}\label{lem:Two}
\begin{enumerate}
\item If $\epsilon_1 \geq \epsilon_2$, then $\OPT_{\epsilon_1} \geq \OPT_{\epsilon_2}$.
\item For every $\alpha \in [0,1]$ and $\epsilon_1,\epsilon_2 > 0$, let $\epsilon = \alpha \epsilon_1 + (1-\alpha)\epsilon_2$, then there exists an $\epsilon$-equilibrium with social welfare 
$\alpha \OPT_{\epsilon_1} + (1-\alpha)\OPT_{\epsilon_2}$.
\end{enumerate}
\end{lem}
\begin{proof}
The first item of the lemma is a trivial observation.
In order to prove the second item, we observe that by Lemma~\ref{lem:PureStrategies} (and since by Proposition~\ref{prop:MixedNotStrongetThenPure} it is enough to consider only pure equilibria) it is enough to prove that there is an infinite path $\pi$ with utility at least $\nu_i - \epsilon$ in every dimension and with social welfare
$\alpha \OPT_{\epsilon_1} + (1-\alpha)\OPT_{\epsilon_2}$.
By Proposition~\ref{prop:MixedNotStrongetThenPure}, it is enough to show that there is a profile $\sigma$ of mixed strategies (that need not be an equilibrium) that has a utility at least $\nu - \epsilon$ in every dimension and has a social welfare at least $\alpha \OPT_{\epsilon_1} + (1-\alpha)\OPT_{\epsilon_2}$.
The construction of $\sigma$ is trivial.
For $i=1,2$, let $\sigma_{\epsilon_i}$ be a profile of strategies that induces an $\epsilon_i$-optimal equilibrium, then we construct $\sigma$ by playing according to $\sigma_{\epsilon_1}$ with probability $\alpha$ and playing according to $\sigma_{\epsilon_2}$ with probability $1-\alpha$.
\pfbox
\end{proof}
\begin{cor}\label{cor:Approx}
For every $\zeta \leq \frac{\epsilon\delta}{4k}$ we have
$\OPT_{\epsilon+\zeta} - \frac{\delta}{2} \leq \OPT_\epsilon \leq  \OPT_{\epsilon+\zeta}$
\end{cor}
By the above corollary, to approximate $\OPT_\epsilon$, it is enough to approximate by $\frac{\delta}{2}$ the value of $\OPT_{\epsilon+\zeta}$ for some $\zeta \leq \frac{\epsilon\delta}{4k}$.
For this purpose, we extend the notion of $\epsilon$-equilibrium also for $k$-dimensional vectors, and we say that a profile of strategies is a $\VEC{\beta}$-equilibrium if player $i$ cannot improve its utility by at least $\beta_i$.
Let us denote by $\min\VEC{\beta}$ and by $\max\VEC{\beta}$ the minimal and maximal element of $\beta$ (respectively).
Then by definition,
$\OPT_{\min\VEC{\beta}} \leq \OPT_{\VEC{\beta}} \leq \OPT_{\max\VEC{\beta}}$,
and by Corollary~\ref{cor:Approx} we get that
$\OPT_{\VEC{\beta}} \leq \OPT_{\max\VEC{\beta}} \leq \OPT_{\VEC{\beta}} + (\max\VEC{\beta}-\min\VEC{\beta})\cdot\frac{4k}{\epsilon}$.

We are now ready to present an FPTAS that computes a $\delta$-approximation for $\OPT_\epsilon$:
(1)~Set $\zeta = \frac{\epsilon\delta}{4k}$, and compute a $\zeta$ approximation of $\nu_i$ for every player $i$, and denote it by $r_i$.
(2)~Compute the optimal path $\pi$ (with respect to social welfare) that has utility at least $r_i - (\epsilon-\zeta)$ for every player, and return its social welfare.

We note that we can execute the first step of the algorithm in polynomial time due to~\cite{mansour}(Observation~3.1), and we can execute the second step in polynomial time by solving the linear programming problem that we described in the proof of Lemma~\ref{lem:MeanPayoffImpliesOptimal}.
The next lemma proves the correctness of our approximation algorithm.
\begin{lem}\label{lem:ApproxIsRight}
Let $S(\pi)$ be the social welfare of $\pi$.
Then $S(\pi)-\delta \leq \OPT_\epsilon \leq S(\pi)$
\end{lem}
Lemma~\ref{lem:ApproxIsRight} along with the complexity analysis that we provided, proves that there is an FPTAS to compute an $\epsilon$-equilibrium that is $\delta$-optimal, and Theorem~\ref{thm:OptEqilib}(4) follows.
We also note that our proof for Theorem~\ref{thm:OptEqilib}(4) gives a constructive (and polynomial) algorithm that computes a description of an actual $\epsilon$-equilibrium that is $\delta$-optimal.
\section*{Acknowledgement}
This work was carried out in partial fulfillment of the requirements for the course Computational Game Theory that was given by Prof. Amos Fiat in Tel Aviv University.
\bibliographystyle{plain}	
\bibliography{CGT}




%
\newpage
\appendix
\section*{Appendix}
\section{Proof of Lemma~\ref{lem:DirectAndUndirect}}
\begin{proof}
We show that for a given directed normalized bipartite complete graph $G=(V,E,w)$ we can construct (in polynomial time) an undirected normalized bipartite complete graph $G'=(V',E',w')$ such that maximizer has a winning strategy in $G$ (for threshold $0$) if and only if he has a winning strategy in $G'$ (again for threshold $0$).
Informally, we build $G'$ from $G$ by taking the same set of vertices and by assigning a weight $w'(u,v)$ that matches either to $w(u,v)$ or to $w(v,u)$.

To formally construct $w'$ we present the notion of \emph{surely loosing edges}.
Recall that $V_1$ are maximizer vertices and $V_2$ are minimizer vertices.
We say that a maximizer vertex out-edge $(v_1,v_2)$ is \emph{surely loosing for maximizer} if $w(v_1,v_2) + w(v_2,v_1) < 0$.
We claim that if maximizer has a winning strategy, then all its memoryless winning strategies do not contain surely loosing edges.
That is, for every memoryless winning strategy $\sigma$ and a surely loosing edge $(v_1,v_2)$ we have $\sigma(v_1) \neq (v_1,v_2)$.
Indeed, towards contradiction let us assume that $\sigma(v_1) = (v_1,v_2)$, and recall that $G$ is a complete bipartite graph, than for the minimizer strategy $\tau$ that leads from every minimizer vertex to $v_1$ we have that the mean-payoff value of $(\sigma,\tau)$ is $\frac{w(v_1,v_2) + w(v_2,v_1)}{2} < 0$, in contradiction to the assumption that $\sigma$ is a winning strategy.
Analogously, $(v_2,v_1)$ is a surely loosing edge for minimizer if
$w(v_1,v_2) + w(v_2,v_1) \geq 0$ and by the same arguments there is no memoryless winning strategy for minimizer that contains a surely loosing edge.

We are now ready to formally define $G'=(V',E',w')$.
We define $V' = V$, and since $G'$ is complete bipartite the definition of $E'$ follows immediately.
For an edge $\{v_1,v_2\}\in E'$, where $v_i\in V_i$, we define 
\[
w'(v_1,v_2) = \left\{ \begin{array}{rl}
w(v_1,v_2) &\mbox{ if $w(v_1,v_2) + w(v_2,v_1) \geq 0$} \\
w(v_2,v_1) &\mbox{ otherwise (if $w(v_1,v_2) + w(v_2,v_1) < 0$)}
\end{array} \right.
\]

We first prove that if $\sigma$ is a memoryless winning strategy for maximizer in $G$, then it is also a winning strategy in $G'$.
Indeed, let $\rho$ be a path (in $G$ and in $G'$) that is consistent with $\sigma$.
We claim that for every path $\pi$ that is a finite prefix of $\rho$ we have
$w(\pi) \leq w'(\pi)$ (that is, its weight in $G'$ is not less than its weight in $G$), and we prove it by a simple induction on the length of $\pi$.
For $|\pi| = 0$ the claim is trivial.
For $|\pi| > 0$ let $(u,v)$ be the last edge in $\pi$.
If $u\in V_1$, then we get that $\sigma(u)=(u,v)$ and therefore $(u,v)$ is not a surely loosing edge.
Therefore, by definition, $w(u,v) = w'(u,v)$ and the claim follows by the induction hypothesis (since it holds for the prefix of length $|\pi|-1$).
If $u\in V_2$, then since $G$ is normalized we have that $w(u,v) < 0$ and $w(v,u) > 0$.
Therefore, by definition, $w'(u,v) \geq w(u,v)$, and by the induction hypothesis the claim follows.
Since $\sigma$ is a winning strategy we get that the mean-payoff value of $\rho$ is non-negative according to $w$, and by the last claim we get that the mean-payoff of $\rho$ according to $w'$ is greater or equal to the mean-payoff value according to $w$.
Hence $\sigma$ is a winning strategy for maximizer also in $G'$.

By the same arguments we get that if minimizer has a memoryless winning strategy in $G$, then the same strategy is winning for minimizer also in $G'$.

Finally, due to Theorem~\ref{thm:OptStrategies}, we get that maximizer has a winning strategy in $G$ if and only if he has a memoryless winning strategy in $G$ if and only if he has a memoryless winning strategy in $G'$ if and only if it has a winning strategy in $G'$ and the proof of the lemma follows. \pfbox
\end{proof}
\section{Proof of Lemma~\ref{lem:PureStrategies}}
\begin{proof}
The direction from left to right is trivial, since a profile of pure strategies has a (unique) infinite path in the graph that is consistent with the strategies.
To prove the converse direction we introduce the notion of a \emph{path strategy} and the notion of \emph{path equilibrium}.
For a path $\pi$ we define the \emph{path strategy} for player $i$ to be:
at round $j$ (which is a player $i$ round), play according to the $j$-th edge in $\pi$.
And we define the strategy $\sigma_i^\pi$ to be:
If player $j\neq i$ deviated from the path strategy of $\pi$, then play forever according to $\overline{\sigma_j}$, and otherwise play according to the path strategy of $\pi$.
The \emph{path equilibrium} of the infinite path $\pi$ is 
the profile $(\sigma_1^\pi,\dots,\sigma_k^\pi)$.
The profile is an equilibrium since it assures a utility $t_i \geq \nu_i$ for every player, and if player $i$ deviates from the strategy $\sigma_i^\pi$ he will end up with a utility $\nu_i$.
\pfbox
\end{proof}
\section{Proof of Lemma~\ref{lem:SufCond}}
\begin{proof}
We assume that for every $\alpha > 0$ there exists a finite cyclic path $C_\alpha$ with average weight at least $t_i - \alpha$ in every dimension, and we let $v_\alpha$ be an arbitrary vertex in the cycle $C_\alpha$.
For every two vertices $u$ and $v$, we denote by $\pi_{u,v}$ the shortest path from $u$ to $v$ (recall that $G$ is strongly connected), and we denote by $W$ the size of the biggest weight (in absolute value) in $G$.
Intuitively, we obtain an infinite path with mean-payoff vector at least $(t_1,\dots,t_n)$ by following the cycle $C_{\alpha}$ for $\alpha = 1$,
and then we follow the path $\pi_{v_\alpha,v_{\frac{\alpha}{2}}}$ and we follow the cycle $C_{\frac{\alpha}{2}}$ twice, and then we follow $C_{\frac{\alpha}{3}}$ three times and so on.
However, we have to make sure that the average payoff does not decrease too much when following a cycle for the first time.

Formally, for every $i\in\Nat$, we denote by $L_i$ the length of the cycle $C_{\frac{1}{i}}$, and by $m_i = iW L_{i+1}$.
We define $\rho_0$ to be the empty path, and for every $i > 0$ we define
$\rho_i = \rho_{i-1} \pi_{v_\frac{1}{i-1},v_\frac{1}{i}} C_{\frac{1}{i}}^{m_i}$,
and we define $\rho$ to be the infinite path that is the limit of the sequence $\rho_0, \rho_1, \dots, \rho_i,\dots$.
Due to the fact that the length of $\pi_{v_\frac{1}{i-1},v_\frac{1}{i}}$ is bounded (by the size of the graph), and since the maximal weight is at most $W$ (and the minimal weight is at least $-W$), we get, by a simple algebra, that the mean-payoff vector of $\rho$ is at least $(t_1,\dots,t_k)$.
\pfbox
\end{proof}

\section{Proof of Lemma~\ref{lem:MixedFollowsCycle}}
\begin{proof}
By definition of the utility function, for every $\delta > 0$, there exists a round $j$ such that the expected average utility, in every round after $j$,  is at least $t_i - \delta$ in every dimension.
We denote by $\Pi_j$ the (finite) set of paths with length $j$ that have non-zero probability according to the strategy profile $\sigma$, and w.l.o.g we assume that all the paths have the same probability (and a path may occur more than once in $\Pi_j$).
For a path $\pi\in \Pi_j$, we denote by $C_\pi$ the longest cyclic path that is a sub-path of $\pi$.
We note that since $G$ is a finite graph, we get that $|C_\pi| \geq |\pi| - 2|G|$ (where $|G|$ denotes the number of vertices in $G$), and in every dimension:
$w(C_\pi) \geq w(\pi) - 2|G|W$ and $\Avg(C_\pi) \geq \Avg(\pi) - \frac{2|G|W}{|\pi|-2|G|W} = \Avg(\pi) - \frac{2|G|W}{j-2|G|W}$.
We partition $\Pi_j$ to $|G|$ sets (some of them may be empty) namely $\Pi_j^v$ for every $v\in G$, such that for every path $\pi \in \Pi_j^v$ we have that $v$ is in the cyclic path $C_\pi$.
For a set $\Pi_j^v = \{\pi_1,\dots,\pi_m\}$, we denote $C(\Pi_j^v) = C_{\pi_1} C_{\pi_2}\dots C_{\pi_m}$ (note that $C(\Pi_j^v)$ is a path).
For every two vertices $u,v \in G$, we denote by $\pi_{u,v}$ the shortest path between $u$ and $v$, and we note that $|\pi_{u,v}| \leq |G|$ and $w(\pi_{u,v})\geq -|G|W$ in every dimension.
Finally, we assume that the vertex set of $G$ is $V = \{v_1,\dots,v_m\}$ and we define the cyclic path
\[\pi = C(\Pi_j^{v_1})\pi_{v_1,v_2}C(\Pi_j^{v_2}) \pi_{v_2,v_3}\dots\pi_{v_{m-1},v_m}C(\Pi_j^{v_m})\pi_{v_m,v_1}\]
The average weight of $\pi$ in every dimension is at least
\[\Avg(\Pi_j) - \frac{3|G|^2W}{j}\]
and since in every dimension $\Avg(\Pi_j) \geq t_i - \delta$,
then for $j \geq \frac{3|G|^2W}{\delta}$ we get that in every dimension
\[\Avg(\pi) \geq t_i - 2\delta\]
and for $\delta = \frac{\alpha}{2}$ we get that $\pi$ is a cyclic path with average weight at least $t_i - \alpha$ in every dimension, and the proof of the lemma follows. \pfbox
\end{proof}
\section{Proof of Lemma~\ref{lem:SocialLowerBound}}
\begin{proof}
Due to Theorem~\ref{thm:DecAndCompTheSame} and Theorem~\ref{thm:Sec1} it is enough to show that if we could approximate the social welfare of the optimal equilibrium in a three-player game, then we would be able to determine whether in a two-player zero-sum game, player 1 has a strategy that assures a value that is strictly greater than $0$.
The proof is straight forward.
Let $(A_1,A_2)$ be the actions of a zero-sum two-player game with utility functions $u_1 : A_1\times A_2 \to [-1,1]$ and $u_2 = -u_1$.
We construct a three-player game $(A_1',A_2',A_3',u_1',u_2',u_3')$ in the following way:
\begin{itemize}
\item $A_1' = A_1 \cup \{\$\}$, $A_2' = A_2 \cup \{\$\}$ and $A_3' =\{\$\}$ (where $\$$ is a fresh action).
\item Let $a_2^*$ be an arbitrary action in $A_2' - \{\$\}$.
We define the utility function $u_1'$ to be
\[
u_1'(a_1,a_2,a_3) = \left\{ \begin{array}{rl}
u_1(a_1,a_2) &\mbox{ if $a_1 \neq \$$ and $a_2 \neq \$$} \\
u_1(a_1,a_2^*) &\mbox{ if $a_1 \neq \$$ and $a_2 = \$$} \\
0 &\mbox{ otherwise (if $a_1 = \$$)}
\end{array} \right.
\]
We define $u_2' = -u_1'$, and 
\[
u_3'(a_1,a_2,a_3) = \left\{ \begin{array}{rl}
1 &\mbox{ if $a_1 = \$$} \\
-1 &\mbox{ otherwise (if $a_1 \neq \$$)}
\end{array} \right.
\]
\end{itemize}
The reader can verify that if player 1 has a strategy to assure utility greater than $0$ in the zero-sum game, then in any profile of equilibrium in the three-player game he will play $\$$ only for a negligible number of rounds, and the social welfare of the equilibrium will be $-1$.
On the other hand, if player 1 cannot assure utility at least $0$, then a profile of strategies in which all three players play $\$$ forever is an equilibrium and its social welfare is $1$.
Hence, even for $\epsilon = 1$ we cannot approximate the social welfare of the optimal equilibrium, unless mean-payoff games are in P. \pfbox
\end{proof}

\section{Proof of Corollary~\ref{cor:Approx}}
\begin{proof}
The fact that $\OPT_\epsilon \leq \OPT_{\epsilon+\zeta}$ follows immediately from Lemma~\ref{lem:Two}(1).
To prove that $\OPT_{\epsilon+\zeta} - \delta \leq \OPT_\epsilon$, we set $\alpha = \frac{\zeta}{\epsilon}$ and by Lemma~\ref{lem:Two}(2), and since $\alpha\zeta + (1-\alpha)(\epsilon+\zeta) = \epsilon$ we get
\[\alpha \OPT_{\zeta} + (1-\alpha)\OPT_{\epsilon+\zeta} \leq \OPT_\epsilon\]
since the utility function is scaled to $[-1,1]$ we get that $\OPT_{\zeta}\geq -k$ and $\OPT_{\epsilon+\zeta} \leq k$, and thus we have
\[-2k\alpha + \OPT_{\epsilon + \zeta} \leq \OPT_\epsilon\]
since $\alpha = \frac{\zeta}{\epsilon}$ we get
\[\frac{-2k\zeta}{\epsilon} + \OPT_{\epsilon + \zeta} \leq \OPT_\epsilon\]
and since $\zeta \leq \frac{\epsilon\delta}{4k}$ we have
\[\OPT_{\epsilon + \zeta} -\frac{\delta}{2} \leq \OPT_\epsilon\]
\pfbox
\end{proof}

\section{Proof of Lemma~\ref{lem:ApproxIsRight}}
\begin{proof}
We denote player-$i$ utility according to $\pi$ by $u_i(\pi)$ and
we construct the vector $\VEC{\beta}$ by defining
\[
\beta_i = \left\{ \begin{array}{rl}
\nu_i - u_i(\pi) &\mbox{ if $\nu_i - u_i(\pi)\geq (\epsilon - \frac{\zeta}{2})$} \\
\epsilon - \frac{\zeta}{2} &\mbox{ otherwise}
\end{array} \right.
\]
and we claim that $\pi$ is an optimal $\VEC{\beta}$-equilibrium.
The claim holds because for every $\VEC{\beta}$-equilibrium, the utility of player $i$ is at least $\nu_i - (\epsilon - \zeta)$, and $\pi$ is the optimal path with utility at least $\nu_i - (\epsilon - \zeta)$ for every player.
In addition, we claim that $\max\VEC{\beta} - \min\VEC{\beta}$ is at most $\frac{\zeta}{2}$.
The claim holds because by the construction of $\beta$ we have that $\min\VEC{\beta}\geq \epsilon - \frac{\zeta}{2}$, and $\nu_i - u_i(\pi)$ is at most $\epsilon - \zeta$ (since $\pi$ is an $(\epsilon - \zeta)$-equilibrium).
Hence $\max\VEC{\beta} - \min\VEC{\beta} \leq \frac{\zeta}{2}$.
Therefore, since $\max\VEC{\beta} \leq \epsilon$ and since
\[\OPT_{\VEC{\beta}} \leq \OPT_{\max\VEC{\beta}} \leq \OPT_{\VEC{\beta}} + (\max\VEC{\beta}-\min\VEC{\beta})\cdot\frac{4k}{\epsilon}\] we get that
\[S(\pi) \leq \OPT_\epsilon \leq S(\pi) + \frac{\zeta}{2}\cdot\frac{4k}{\epsilon\delta}\]
and since $\zeta = \frac{\epsilon\delta}{4k}$ we get that
\[S(\pi) \leq \OPT_\epsilon \leq S(\pi) + \frac{\delta}{2}\]
and the assertion of the lemma follows. \pfbox
\end{proof}

\end{document}